\documentclass{llncs}

\usepackage{amsmath}
\usepackage{amssymb}
\usepackage{wrapfig}
\usepackage{float}
\usepackage{graphicx}
\usepackage{caption}
\usepackage{subfloat}
\usepackage{algpseudocode}
\usepackage{algorithm}
\usepackage{etoolbox}
\usepackage{amssymb}
\usepackage{tikz}
\usepackage{mathtools}

\title{Succint greedy routing without metric on planar triangulations}
\author{Pierre Leone, Kasun Samarasinghe}
\institute{Computer Science Department, University of Geneva, Battelle A, route de Drize 7, 1227 Carouge, Switzerland}

\begin{document}

\maketitle

\begin{abstract}
Geographic routing is an appealing routing strategy that uses the location information of the nodes to route the data. This technique uses only local information of the communication graph topology and does not require computational effort to build routing table or equivalent data structures. A particularly efficient implementation of this paradigm is {\it greedy routing}, where along the data path the nodes forward the data to a neighboring node that is closer to the destination. The decreasing distance to the destination implies the success of the routing scheme. A related problem is to consider an abstract graph and decide whether there exists an embedding of the graph in a metric space, called a {\it greedy embedding}, such that greedy routing guarantees the delivery of the data. In the present paper, we use a metric-free definition of greedy path and we show that greedy routing is successful on planar triangulations without considering the existence of greedy embedding. Our algorithm rely entirely on the combinatorial description of the graph structure and the coordinate system requires ${\cal O}\bigl(log(n)\bigr)$ bits where $n$ is the number of nodes in the graph. Previous works on greedy routing make use of the embedding to route the data. In particular, in our framework, it is known that there exists an embedding of planar triangulations such that greedy routing guarantees the delivery of data. The result presented in this article leads to the question whether the success of (any) greedy routing strategy is always coupled with the existence of a greedy embedding? 

\end{abstract}


\section{Introduction}
Geometric routing is an appealing routing technique that uses the position of the nodes for routing data in communication networks. In particular, {\it greedy routing} consists in using routing paths such that at each hop the distance to the destination decreases. Greedy routing is local. Unfortunately, there are simple examples where all neighboring nodes of a node are at larger distance to the destination and greedy cannot be applied. It is then relevant to determine in which situations greedy routing guarantees the delivery of data. Besides the relevance to routing data in communication networks, greedy routing has lead to interesting works in graph theory and computational geometry.

The paper \cite{DBLP:conf/algosensors/PapadimitriouR04} is usually considered as a first landmark. In particular,  it is shown how a $3$-connected graph can embedded in ${\mathbb R}^3$ in such a way that greedy routing guarantees delivery. Moreover, it is {\it conjectured}  that any $3$-connected planar graph can be embedded in $\mathbb R^2$ such that greedy routing guarantees delivery, what is now called a {\it greedy embedding} (in $\mathbb R^2$). This conjecture leads to intensive research and an exhaustive survey of all the results is not presented here. We mention key contributions in the main directions around this conjecture and position the contribution of this article among these works.

The conjecture was proved to be true. In \cite{DBLP:conf/focs/MoitraL08,eppstein2011succinct} it is proved for $3$-connected graphs, in \cite{bose1999online} in Delaunay triangulations, in \cite{chen2007distributed} for graphs that satisfy conditions with respect to the power diagram. In \cite{DBLP:conf/infocom/Kleinberg07}, a greedy embedding in the hyperbolic plane of a connected finite graphs is constructed. More related to our approach, in \cite{Dhandapani08} the conjecture is proved for planar triangulations (maximal planar graphs), see also \cite{he2010schnyder,he2011succinct}. All these approaches rely on a distance metric that is used by greedy routing. The papers \cite{eppstein2011succinct,he2010schnyder,he2011succinct} limit the memory requirement to ${\cal O}\bigl(log(n)\bigr)$ bits to represent the coordinates. Such coordinate systems are called {\it succint} and this property is important for the design of scalable routing scheme. Our coordinate system, see definition \ref{def:coordinates}, is {\it succint}. In \cite{Dhandapani08,he2010schnyder,he2011succinct}, Schnyder's caracterization of (maximal) planar graphs \cite{schnyder} is used. This characterization is discussed in section \ref{vrac} and is also used in this work. However, our approach is to avoid the definition of a metric and the computation of the planar embedding, see for instance \cite{angelini2009algorithm} for an algorithm to compute the greedy embedding of planar triangulations. We rely on the metric-free definition of greedy paths in \cite{DBLP:journals/tmc/LiYL10} - without embedding the graph. Moreover, the coordinate system used in \cite{Dhandapani08,he2010schnyder,he2011succinct} is more complex to compute than the one we use. Briefly, in Schnyder work, planar graphs are characterized by the existence of three total order relations on the vertex set of the graph (and extra conditions). Using these order relations Schnyder builds three spanning\footnote{ Actually, spanning internal nodes of the triangulation.} directed trees called the {\it realizer} and the coordinates are computed using these trees. This coordinate system has relevant properties, for graph drawing, see for instance \cite{nishizeki2004planar}, that we do not need for routing\footnote{Although it is relevant to ask if these properties are necessary for constructing a greedy embedding.}. The particularities of greedy routing with respect to greedy drawing are already pointed out in \cite{eppstein2009succinct}.

In this article, we use directly the order ranks for coordinates, see Definition \ref{def:coordinates} without (non-local) extra-computations. This approach is coherent with previous works. In particular \cite{DBLP:conf/HucJLR12,DBLP:conf/algosensors/HucJLR10} where the three order relations are obtained by the nodes composing a communication network by measuring the distances to three distinguished nodes, see Figure \ref{fig:exvrac}. The motivation for this approach to avoid the computation of the coordinates, i.e. the usual scheme for localizing the nodes is $1)$ measure distances and, 2) compute the coordinates. In the subsequent works \cite{DBLP:conf/icpads/SamarasingheL12,DBLP:conf/sensornets/SamarasingheL14}, we design the classical {\it greedy-face} routing paradigm in this simple coordinate system.

We emphasize that the underlying assumptions in these works (or generally in geographic routing) is that the communication graph is a Unit Disk Graph (UDG) and the nodes are located in a $2D$ region, in particular we can use the Jordan curve theorem to prove that face routing guarantees the delivery. In this work, we do not make such assumption. We assume that the graph is given by the sets of nodes and edges in an abstract way.

\noindent{\bf Our contribution} In this paper we show that given a maximal planar graph we can
\begin{itemize} 
\item Provide a metric-free definition of greedy paths.
\item Design a local greedy routing algorithm that guarantees delivery.
\item Decouple the problems of greedy routing and greedy embedding.
\end{itemize}
Moreover, all our computations are constructive and the coordinate system requires ${\cal O}\bigl(log(n)\bigr)$ bits and can be qualified of {\it succint} \cite{eppstein2011succinct}.

In section \ref{vrac} we present Schnyder's characterization of planar maximal graphs, fix the notations and define greedy paths. In section \ref{greedy} we prove the properties that we need to build a greedy path between any two nodes.  Finally, in section \ref{routingmaxplanargraph} we state the main result of the paper about the existence of greedy paths.

\section{Schnyder three-dimensional representation and coordinates}
\label{vrac}

Given a planar graph $G=(V,E)$, it is proved in \cite{schnyder} that there exists three total order relations on $V\times V$, denoted $<_1, <_2, <_3$ such that
\begin{equation}
\begin{array}{l l}\label{eq:schnydercond}
&a)~~\bigcap_{i=1,2,3} <_i =\emptyset,\text{, and}\\
&b)~~\forall (x,y)\in E,  \forall z\not\in\{x,y\}~ \exists i\in\{1,2,3\} \text{ s.t. } (x,z)\in <_i \text{ and } (y,z)\in <_i.
\end{array}
\end{equation}
This is called a {\it (3-dimensional) representation} of the planar graph.
We also use the notation $x<_iz$ for $(x,z)\in <_i$ and we say $v$ is a neighboring node of $u$ to say that $(u,v)\in E$. 

\begin{definition}(standard representation, internal and external nodes)\label{def:internal} The representation is {\it standard} if there are three distinguished elements $A_1, A_2, A_3$ such that $A_i$ is the maximal element for $<_i$ and the remaining $A_{i-1}, A_{i+1}$\footnote{We use the notation $ i+1\equiv i \text{ mod } 3 +1$.} are the two smallest elements of $<_i$ (in any order, $A_{i-1}<_1 A_{i+1}$ or $A_{i+1}<_1 A_{i-1}$). Any representation can be turned to a standard one \cite{schnyder}. 

 The distinguished elements $A_i, i=1,2,3$ of a standard representation are called {\it external}, the others nodes are called {\it internal}.
\end{definition}

\begin{definition}\label{def:coordinates}(Coordinates) Given a node $u\in V$ the rank of $u$ in the order $<_i$ is denoted by $rank_i(u)$. If there are $n$ elements in $V$,  the smallest element has rank $1$ and the largest rank $n$. The coordinates of the node $u$ are $(u_1,u_2,u_3)=(rank_1(u), rank_2(u), rank_3(u))$. This coordinate system requires 3 integers per nodes and hence, it scales like ${\cal O}\bigl(log(n)\bigr)$ bits and is {\it succint} \cite{eppstein2011succinct,he2010schnyder,he2011succinct}.
\end{definition}

\begin{figure}[!]\hskip -2cm
\includegraphics[scale=0.3]{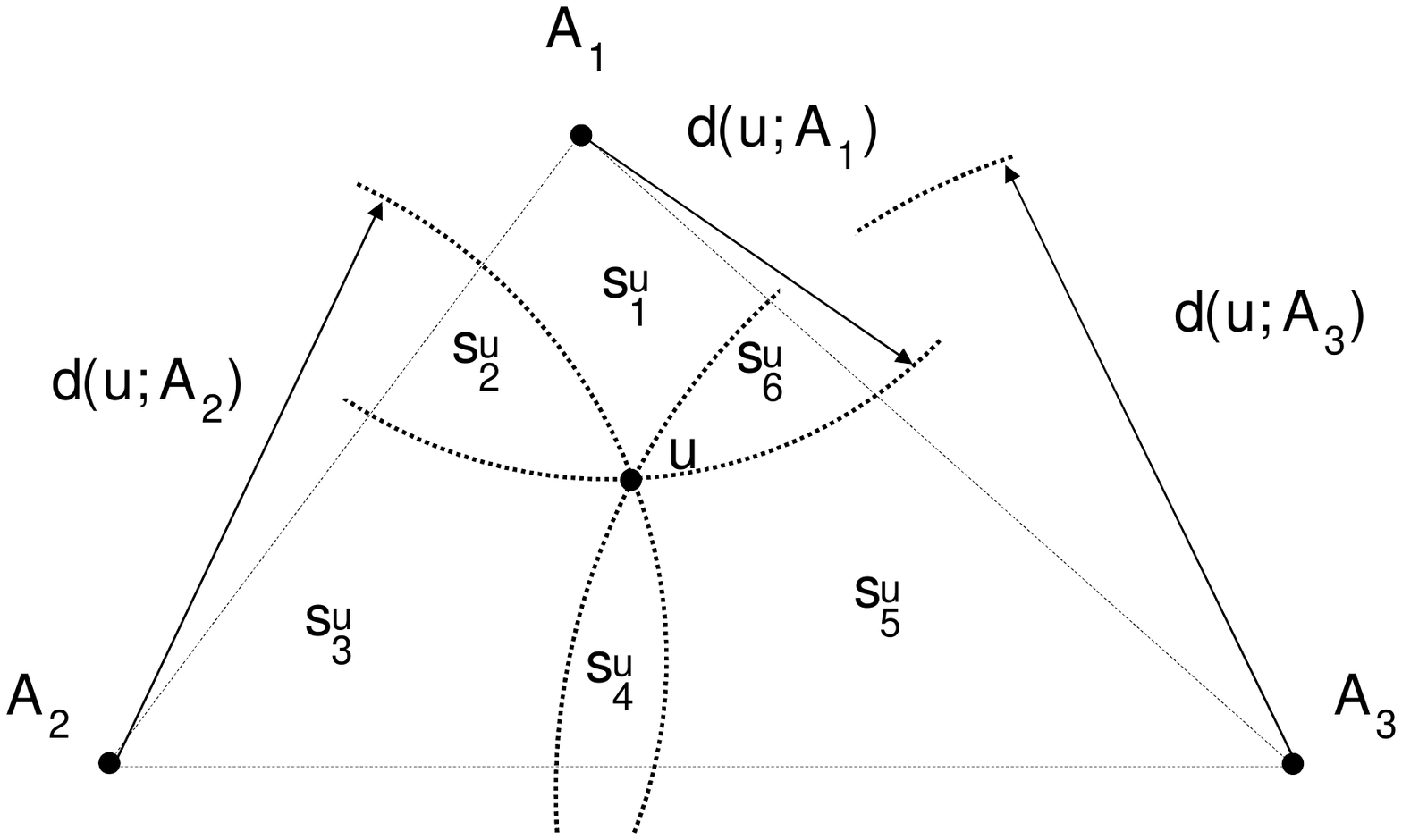}\hskip -1cm\includegraphics[scale=0.3]{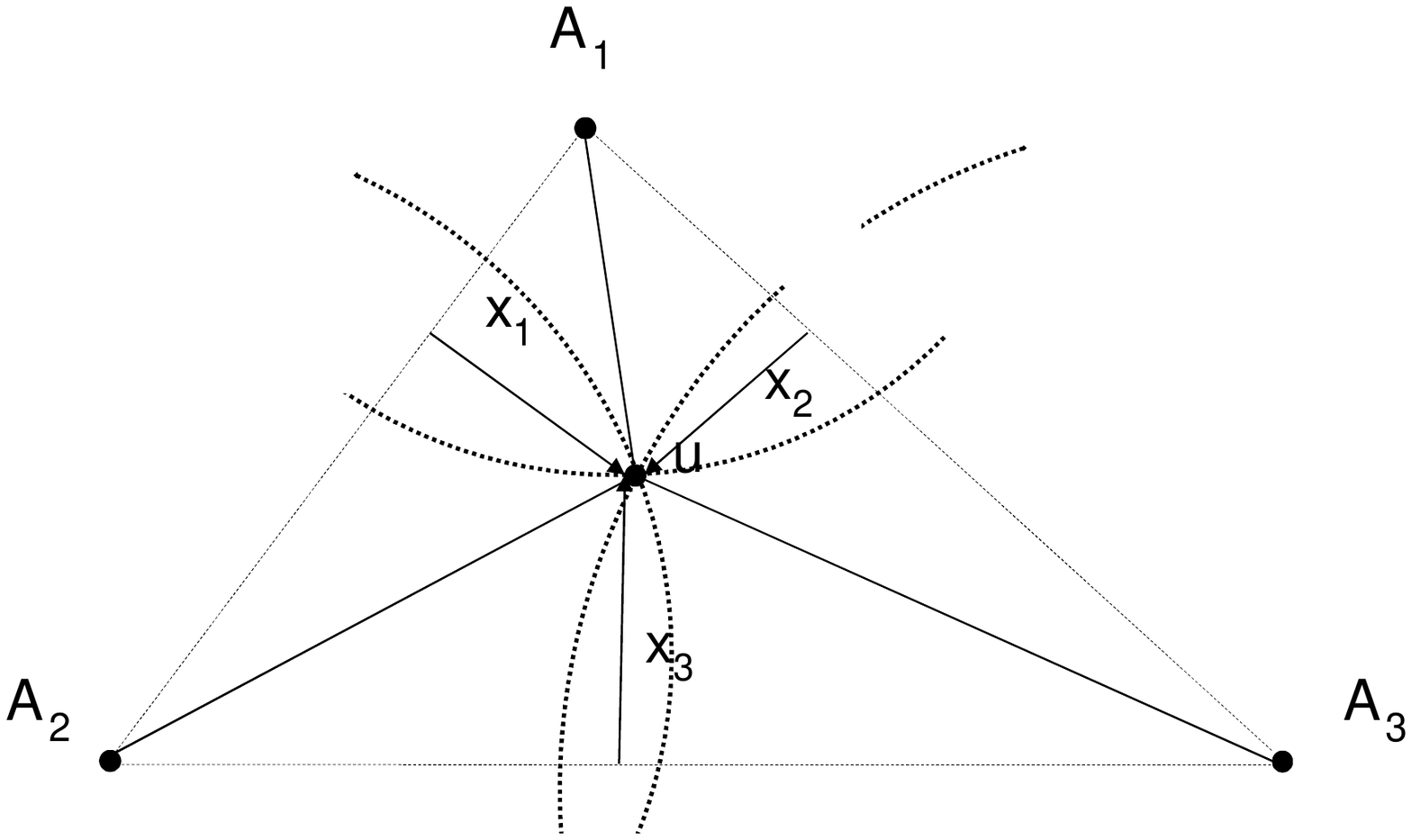}\vskip -2cm
\caption{Two realistic way of getting the order relations $<_i$. Coordinate assignment with with raw distances from anchors on the left and  perpendicular distances from heights of the triangles on the right}
\label{fig:VRAC2}\label{fig:VRAC}\label{fig:exvrac}
 \end{figure}

In ~\cite{DBLP:conf/sensornets/SamarasingheL14} the three order relations are obtained by measuring the distances from the node $u$ to three anchors $d(u,A_i)$ or,  in  \cite{DBLP:conf/icpads/SamarasingheL12} by the three heights of the triangles (needs more computation and information) $uA_iA_j$, see Figure \ref{fig:exvrac}. The order relations are defined by $u<_i v \Leftrightarrow d(u,A_i)>d(v,A_i)$. The motivation for using this VRAC coordinate system \cite{DBLP:conf/algosensors/HucJLR10}, in ad-hoc communication networks, is to avoid the usual computations for the  localization of the nodes. Indeed, most localization schemes start to measure distances and then proceed to computations. In the VRAC coordinate system the computations are avoided. Although in this article we do not assume that the order relations are given in this way, it is helpful to keep these representations in  mind. To help visualize the definitions and results we show how they can be represented if such a geometrical model is assumed. We emphasize that the goal of the article is the design of greedy routing algorithm without geometry.

The three order relations are total\footnote{A total order is a binary relation which is valid for \textit{all the pairs} in a set} and it makes sense to associate the minimum of a set with respect to one of the three order. We will denote this by $\text{min}_i$ for $i=1,2,3$.

These three orders permit the definition of sectors associated with a node $u$.\\
\begin{definition}(Sectors) \label{def:sectors} We define the following sectors associated to a node $u\in V$, see Figure \ref{fig:VRAC}. Note that the reference node $u$ does not belong to the sectors.
\begin{align*}
s_1^u = \{v\mid ~u~<_ 1 ~v, ~u~>_ 2 ~v, ~u~>_ 3 ~v\}.\\
s_2^u = \{v\mid ~u~<_ 1 ~v, ~u~<_ 2 ~v, ~u~>_ 3 ~v\}.\\
s_3^u = \{v\mid ~u~>_ 1 ~v, ~u~<_ 2 ~v, ~u~>_ 3 ~v\}.\\
s_4^u = \{v\mid ~u~>_ 1 ~v, ~u~<_ 2 ~v, ~u~<_ 3 ~v\}.\\
s_5^u = \{v\mid ~u~>_ 1 ~v, ~u~>_ 2 ~v, ~u~<_ 3 ~v\}.\\
s_6^u = \{v\mid ~u~<_ 1 ~v, ~u~>_ 2 ~v, ~u~<_ 3 ~v\}.\\
\end{align*}

\noindent Notice that the coordinates of the nodes in Definition \ref{def:coordinates} make possible to determine in which sector a node belongs relatively to another one. Sometimes, the sector $s^u_i$ is also referred to as {\it the sector $i$ of $u$}\footnote{We use the notation $ i+1\equiv i \text{ mod } 6 +1$ if $i$ is the index of a sector, i.e. $s^u_i$.}.
\end{definition}

\begin{definition}\label{def:sectornotation}
Given a node $D$, we also use the convenient notation $s^u_D$ to denote the sector $j$ of $u$ such that $D\in s^u_j$, i.e. $D\in s^u_D$.
\end{definition}

There is a useful way to distinguish the edges that uses the definition of partial orders $<_1^*, <_2^*, <_3^*$, see Lemma 3.1 in \cite{schnyder}
\begin{definition} For each $i\in\{1,2,3\}$ we define
\begin{equation*}
(u,v)\in <_i^*~ \Longleftrightarrow~ (u,v)\in<_i \text{ and } (v,u)\in<_{i+1} \text{ and } (v,u)\in<_{i-1}.
\end{equation*}
Or equivalently
\begin{equation*}
(u,v)\in <_i^*~ \Longleftrightarrow~  v\in s^u_{2i-1}.
\end{equation*}
\end{definition}
\begin{property}\label{prop:order*}
The empty intersection property $a)$  in $(\ref{eq:schnydercond})$ implies that for each $u,v\in V$ there exists exactly one $i\in{1,2,3}$ such that $(u,v)\in <_i^*$ or $(v,u)\in <_i^*$ (equivalently $v\in s^u_{2i-1}$ or $u\in s^v_{2i-1}$). It is convenient to rememember that if $(u,v)<_i^*$ then $v\in s^u_{2i-1}$, i.e. $s^u_1$ or $s^u_3$ or $s^u_5$, the indexes are odd and even otherwise.
\end{property}

\begin{property}\label{prop:sectors}
A node $u$ has at most one edge $(u,v)\in E$ such that $v\in s_{2i-1}^u$. Moreover,  such a node $v$ satisfies that $v<_i z ~\forall z\in s_{2i-1}^u, i=1,2,3$, i.e. $v=min_i\{z\mid z\in s^u_{2i-1}\}$,  see Lemma 3.1 of \cite{schnyder},  this follow from condition $b)$ in $(\ref{eq:schnydercond})$. These properties can be written
\begin{align}
\left.
{\bf\text{If } v\in s^u_1,~z\not=v\text{  we have }}~~\begin{array}{l l}
z<_2 u\\
z<_3 u
\end{array}\right\} \Rightarrow z>_1 v.\label{eq:propertesu1}\\
\left.
{\bf\text{If } v\in s^u_3,~z\not=v\text{  we have }}~~\begin{array}{l l}
z<_1 u\\
z<_3 u
\end{array}\right\} \Rightarrow z>_2 v.\label{eq:propertesu3}\\
\left.
{\bf\text{If } v\in s^u_5,~z\not=v\text{  we have }}~~\begin{array}{l l}
z<_1 u\\
z<_2 u
\end{array}\right\} \Rightarrow z>_3 v.\label{eq:propertesu5}
\end{align}
\end{property}
Property \ref{prop:sectors} is from \cite{schnyder} and the proof uses part $ b)$ of the graph representation $(\ref{eq:schnydercond})$. There is a nice geometric {\it void} condition associated to this property. Indeed, if we assume that the edge $(u,v)$ belongs to $s^u_1$ then the existence of a node $ u<_1 w <_1 v$ violates the second condition of $(\ref{eq:schnydercond})$, see Figure \ref{fig:voidregion} and Figure 3c of \cite{Dhandapani08}. It is interesting to compare this void region with the corresponding ones of the planar Relative Neighborhood Graphs (RNG) or Gabriel Graphs (GG) \cite{DBLP:conf/mobicom/KarpK00}.

\begin{figure}\hskip-1cm
\includegraphics[scale=0.3]{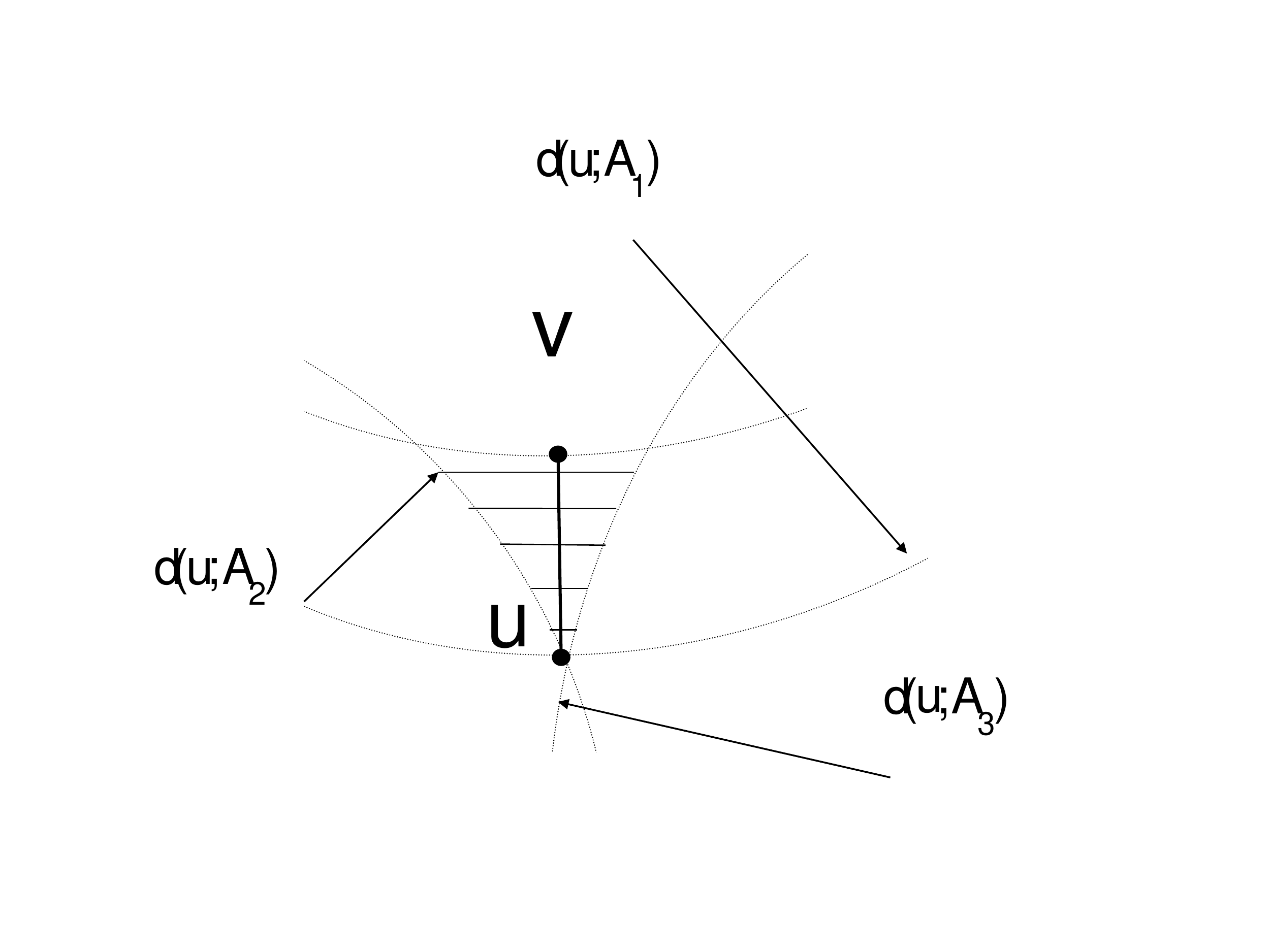}\hskip -1cm\includegraphics[scale=0.3]{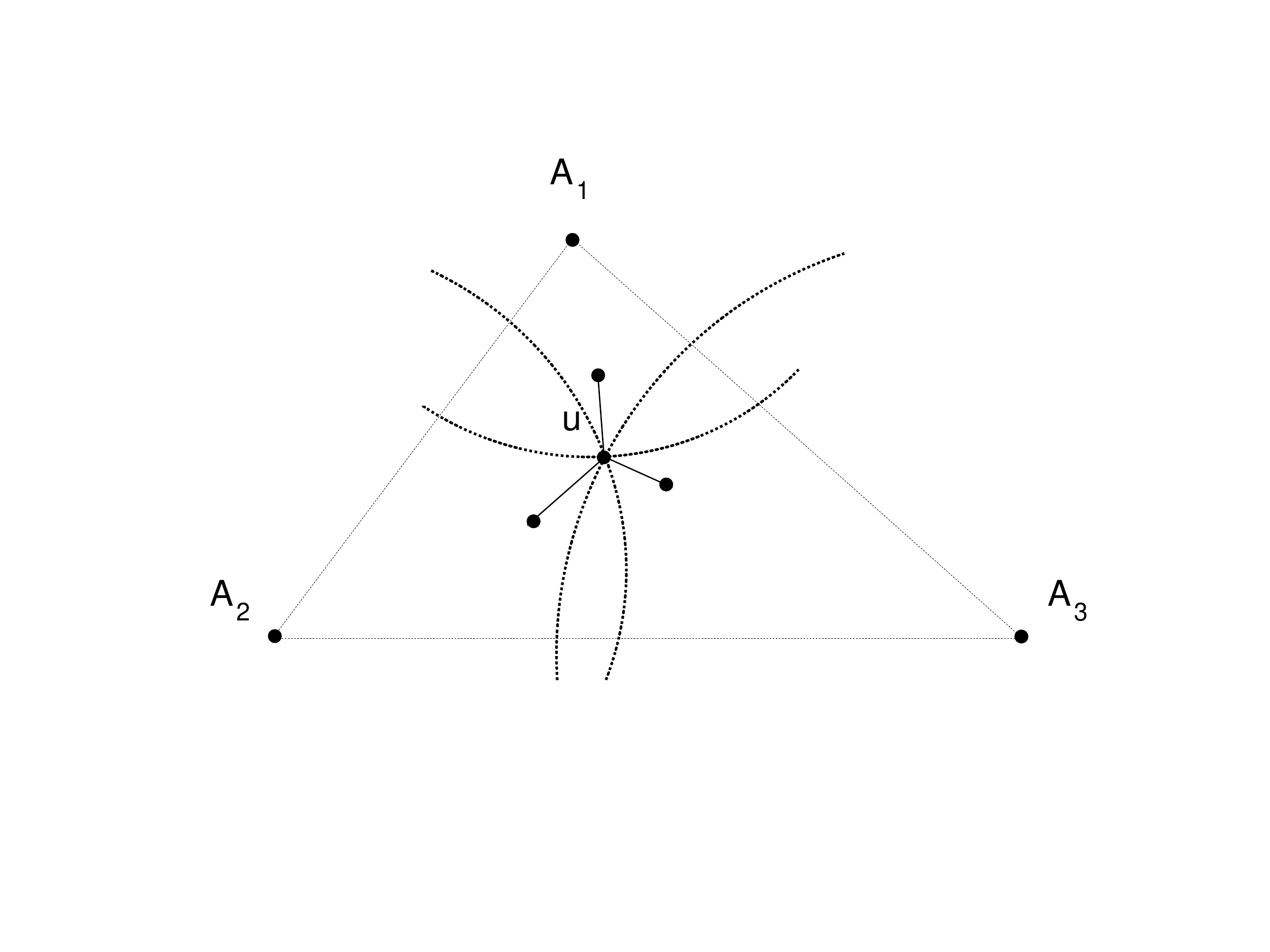}\caption{{\bf Left: }The geometric interpretation of Property \ref{prop:sectors}. The hatched region is void. {\bf Right:}  Locally a node $u$ of a maximal planar graph as exactly three edges in the sectors $s_1^u, s_3^u, s^5_u$ if the representation is standard. The undetermined (0,1,2,...) number of edges in the sectors $s^u_2, s^u_4, s^u_6$ are not represented. }\label{fig:localview}\label{fig:voidregion}
\end{figure}
\subsection*{Informal presentation of the routing strategy}

\noindent If the graph is planar maximal  and the representation standard then each internal node $u$, see Definition \ref{def:internal},  has exactly one edge in each sector $s^u_1, s^u_3, s^u_5$ and an indeterminate ($0,1,2,\ldots$) number in the remaining sectors $s^u_2, s^u_4, s^u_6$. Indeed, because the representation is standard the sectors $s^u_1, s^u_3, s^u_5$ contain the nodes $A_1, A_2, A_3$ respectively and are not empty. The maximality assumption implies that if there is an option of adding an edge and keeping the planarity property then the edge is present \cite{schnyder}. It is helpful to look at the geometric visualization in the right of Figure \ref{fig:localview}.

For routing from a node $u$ to a destination $D\in s^u_1\cup s^u_3\cup s^u_5$ the natural option is to follow the edge $(u,v)$ such that $v\in s^u_D$ ($=s^u_1$ or $s^u_3$ or $s^u_5$). Next, from $v$, if $D\in s^v_1\cup s^v_3\cup s^v_5$ we repeat the same strategy. However, it may happen that $D\not\in s^v_1\cup s^v_3\cup s^v_5$, see Proposition \ref{lemmagreedy}.  In this case $D\in  s^v_2\cup s^v_4\cup s^v_6$ and the existence of an edge in the sector $s^u_D$ is not provided by the Schnyder's characterization $(\ref{eq:schnydercond})$. Nevertheless, in Proposition \ref{prop:greedyeven} we show how we can route the data in this case.

\section{Greedy Routing}
\label{greedy}

Our greedy routing technique differs from the classical ones as we do not assume that a metric is given. Instead, we use  the metric-free axioms characterizing greedy paths provided in \cite{DBLP:journals/tmc/LiYL10}, i.e. given a destination $D$ we have 1. ({\bf transitivity})  if node $v$ is greedy for $u$ and $w$ is greedy for $v$ then $w$ is greedy for $u$ as well and  2. ({\bf odd symmetry}) if $v$ is greedy for $u$ then $u$ is not greedy for $v$. Moreover, the coordinate system that we use is different than the one used in the others works that are based on Schnyder's characterization of planar graph, for example  in \cite{DBLP:dblp_journals/tcs/HeZ13,Dhandapani08}. Indeed, they use a coordinate system used in \cite{schnyder} that is more complex to compute than our in Definition \ref{def:coordinates}. This difference is possible here because initially the coordinate system was designed to draw the planar graph, while we limit our purpose to route the data.

\begin{definition} For destination node $D$, a path $\{u^k\}$ is a greedy path if there exists $i\in\{1,2,3\}$ such that
\begin{equation}\label{eq:shrinkingregion}
\forall k~~ u^{k+1}<_i{u^k},\text{ or } \forall k~~ u^{k+1}>_i{u^k}.
\end{equation}
For a greedy path there is a coordinate that changes monotonically. 
\end{definition}
Because the coordinates change by at least one unit along a path, a  greedy path must stop. In the following we build greedy paths from $u$ to $D$ such that $u<_i u^k<_iD$ the fact that $D$ is an upper bound and the construction continues while $u^k<_i D$ implies the convergence of the sequence to $D$.
\remark In the proofs of Propositions \ref{prop:greedyeven} and \ref{lemmagreedy} we use the assumption that the graph is maximal to say that given a node $u$ there exists neighboring nodes in the sectors $s^u_1, s^u_3, s^u_5$. Unfortunately we must proceed with caution if the node $u$ is one of the distinguished nodes $A_1, A_2, A_3$ since these nodes may not have any neighboring nodes in these sectors. Actually, these nodes do not cause any trouble because there is a path from any internal nodes to them with increasing coordinate $rank_1, rank_2, rank_3$ respectively. They are also all connected to each others. For these reasons and in order to make our best to simplify the exposition we no longer make any reference to these particular nodes in the proofs.

In the proof of Proposition \ref{prop:greedyeven} we need the following Proposition
\begin{proposition}\label{lemmatransitivity} $D'\in s^D_i$ and $D''\in s^{D'}_i$ then $D''\in s^D_i$
\end{proposition}
\begin{proof}
This property follows directly from the transitivity of the inequalities in the definition of the sectors $(\ref{def:sectors})$.
\end{proof}

\begin{proposition}\label{prop:greedyeven} We assume that the graph $G$ is triangular (or equivalently planar maximal, this implies the existence of a unique edge in each sectors $s^u_1, s^u_3, s^u_5$ of $u$ for all $u$ internal nodes). Then provided that the destination $D$ belongs to $s^u_2$  (or $s^u_4$, or $s^u_6$) then there is a path  $\{u^i\}$ in $G$ with $u^0=u$ such that $u^{i+1} \in s_2^{u^i}$ ($u^{i+1} \in s_4^{u^i}$ or $u^{i+1} \in s_6^{u^i}$ respectively), and the path converges to $D$.

Along the path the coordinate $rank_3$ ($rank_1$, $rank_2$) decreases monotonically if $D\in s^u_2$ ($D\in s^u_4$, $D\in s^u_6$ respectively).
\end{proposition}
\begin{proof}
For concreteness we consider $D\in s^u_4$. If $u$ is connected to $D$ we define $u^1=D$ and the proposition is true. Else, we prove below that there exists a neighboring node of $u$, $u^1$ such that $D\in s^{u^1}_4$ and $D<_1 u^1<_1 u$. Hence, by applying the construction iteratively we construct the sequence of points that satisfy $u^{i+1}\in s^{u^i}_4$, lower bounded by $D$ and decreases with respect to $<_1$, i.e. $D<_1 u^{i+1}<_1 u^i$. Such a sequence converges to $D$.


{\bf Let us prove that given $u$ such that $D\in s^u_4$ there exists $v$ such that $(u,v)\in E$, $D\in s^v_4$ and $D<_1 v<_1 u$. } $u$ is internal, by the maximality (triangulation) assumption there exists two neighboring nodes of $u$ such that $v\in s^u_3$ and $w\in s^u_5$. we then have
\begin{align}
&D <_1 u, ~D >_2 u,~D>_3 u~\Leftrightarrow D\in s^u_4\label{eqq1}\\
&v<_1 u,~v>_2 u,~v<_3 u~\Leftrightarrow v\in s^u_3\label{eqq2}\\
&w<_1 u,~w<_2 u,~w>_3 u~\Leftrightarrow w\in s^u_5\label{eqq3}
\end{align}
If $v$ (or $w$) is such that $D\in s^v_4$ (or $D\in s^w_4$) the next point on the path is $u^1=v$ (or $u^1=w$) and $(\ref{eqq2})$ shows that $v=u^1<_1 u$, and $D\in s^v_4\Rightarrow v>_1 D$ (or $(\ref{eqq3})$ shows that $w=u^1<_1 u$ , and $D\in s^w_4\Rightarrow w>_1 D$). 

\noindent Else, we have to prove that there exists a neighboring node of $u$ in the sector $s^u_4$ that satisfies the conditions. We have that $D>_2u>_2w$, and $D>_3u>_3v$ (using $(\ref{eqq1},\ref{eqq2},\ref{eqq3})$) and  $D\not\in s^v_4$ and $D\not\in s^w_4$ imply
\begin{align}
&D\not\in s^v_4\Rightarrow\left.\begin{array}{l l} D>_1v ~D<_2 v~D>_3v~~\text{ or}\\ D<_1v~D<_2v~D>_3v\end{array}\right\}\Rightarrow D<_2v\label{eq:notinsv4}\\
&D\not\in s^w_4\Rightarrow\left.\begin{array}{l l} D>_1w ~D>_2w~D<_3w~~\text{ or}\\ D<_1w~D>_2w~D<_3w\end{array}\right\}\Rightarrow D<_3w\label{eq:notinsw4}
\end{align}
Next, because $D\in s^u_4\Rightarrow u\in s^D_1$ and the maximality assumption, there exists an edge $(D,D')$ with $D'\in s^D_1$. If $D'=u$ we are done.

\noindent Else we have by the property $(\ref{eq:propertesu1})$ and $u\in s^D_1$ that $\boxed{D'<_1 u}$. 

By gathering the inequalities corresponding to $u\in s^D_1$ with the ones deduced from $(\ref{eq:notinsv4})$,$(\ref{eq:notinsw4})$ we obtain  $D<_1D', v>_2D>_2 D', w>_3D>_3D'$. Using $D' <_1u$, $D'<_2v$ with property $(\ref{eq:propertesu3})$ we obtain $\boxed{D' >_3 u}$. 

Last from $D'<_1 u$, $D' <_3 w$ and property $(\ref{eq:propertesu5})$ (with edge $(u,w)$ instead of $(u,v)$) we obtain $\boxed{D' >_2 u}$.
Finally, we have proved that $D'\in s^u_4$ with the boxes equations and $D<_1D' <_1 u$. The node $D'$ plays the same role as $D$ in the statement of the proposition but with an increasing $<_1$ order position. Because of the bound $D' <_1 u$ we see that by applying iteratively the construction we obtain a sequence $D', D'', \ldots$ that converges to $u$ and such that all the points belong to $s^u_4$. Moreover, along the sequence we have $D'\in s^{D}_1$,  $D''\in s^{D'}_1$,... and Lemma \ref{lemmatransitivity} implies that all the points in the sequence belong to $s^D_1$. In particular, for the point $x$ that is connected to $u$ $x\in s^D_1\Leftrightarrow D\in s^x_4$. We have then proved the existence of a point $x\in s^u_4$ that satisfies $D\in s^x_4$ and such that $D<_1x<_1u$.
\end{proof}

\begin{remark}{\bf Construction of the greedy path if $D\in s^u_{2i}$}\label{remarkeven}

\noindent In order to route from $u$ to $D\in s^u_{2i}$ the node $u$ must first check whether for $v\in s^u_{2i+1}$ and $w\in s^u_{2i-1}$ one of the condition $D\in s^v_{2i}$ or $D\in s^w_{2i}$ is satisfied and if yes sends the message accordingly. Otherwise, the message is forwarded to (the existing) neighboring node in $x\in s^u_{2i}$ such that $D\in s^x_{2i}$. This routing scheme converges because the coordinate $i$ decreases along the path and the path don't step over $D$ because all the points in the path are $>_1 D$.

\end{remark}

\begin{proposition}\label{lemmagreedy} Let us assume that $(u,v)\in E$ and $D,v\in s^u_1$ (or $s^u_3$ or $ s^u_5$). Then, $D\not\in s^v_3\cup s^v_4\cup s^v_5$ (or $s^v_1\cup s^v_5\cup s^v_6$ or $s^v_1\cup s^v_2\cup s^v_3$).
\end{proposition}
\begin{proof} Let us consider $v,D\in ^u_1$ the other cases are proved similarly by a permutation of the indices. We have
\begin{equation*}
\begin{array}{l l}
v\in s^u_1\Leftrightarrow u<_1 v~~u>_2 v~~u>_3 v\\
D\in s^u_1\Leftrightarrow u<_1 D~~u>_2 D~~u>_3 D.
\end{array}
\end{equation*}
Part b) of the Schnyder's conditions $(\ref{eq:schnydercond})$ implies that $D$ must be larger than $u$ and $v$ for one order and we see on the two inequalities above that it can  only be $<_1$. The condition $D\in s^v_3\cup s^v_4\cup s^v_5$ implies that $v>_1 D$ and hence there is no $i\in\{1,2,3\}$ such that $u,v<_i D$ and the result in proved.
%
\end{proof}

\begin{remark}{\bf Construction of the greedy path if $D\in s^u_{2i-1}$}\label{remarkodd}

\noindent The practical implication of Proposition \ref{lemmagreedy} for routing is to prove the existence of a greedy path from $u$ to $D\in s_{2i-1}^u$. We decompose the construction in two parts and for concreteness we consider $D\in s^u_1$.

\noindent {\bf Part 1.}  The maximality assumption implies the existence of a node  $v\in s^u_1$ such that $(u,v)\in E$. If $v=D$ we are done. Else, $u$ sends the message to $v$ and the first coordinate $rank_1$  increases, the second one $rank_2$ and the third one $rank_3$  decrease. If $D\in s^v_1$ then $v$ repeats the same procedure and the coordinates continue to be updated monotonically and $D>_1 v$ because $D\in s^v_1$ and this implies that the first part of the construction converges to $D$ or switches to the second part. 

\noindent {\bf Part 2.} If the path reaches a node $v$ such that $D\not\in s^v_1$ the construction of the path continue with this second part. In this case $D\in s^v_2$ or $D\in s^v_6$ must be satisfied because of Proposition \ref{lemmagreedy}. In both cases we have $D>_1 v$ and we can apply Proposition \ref{prop:greedyeven} that shows the existence of a sequence of nodes $v'$ with $D\in s^{v'}_2$ or  $D\in s^{v'}_6$ respectively and this sequence eventually reaches $D$. If  $D\in s^{v'}_2$ then by Proposition \ref{prop:greedyeven} the coordinate $rank_3$ continues to decrease along the second part of the construction. If $D\in s^{v'}_6$ the coordinate $rank_2$ continues to decrease. In both cases we have shown that along the two parts of the construction one coordinate ($rank_2$ or $rank_3$) decreases monotonically and the resulting path is then {\it greedy}.
\end{remark}

\section{Routing in maximal planar graph}\label{routingmaxplanargraph}
In \cite{Dhandapani08}, it is proved using Schnyder's characterization of planar graphs $(\ref{eq:schnydercond})$ that there exists an embedding of the graph in the plane\footnote{Actually in the plane in $\mathbb R^3$ such that $x+y+z=1$.} such that greedy routing is successful (using the natural metric). In \cite{DBLP:dblp_journals/tcs/HeZ13} the authors use a similar coordinate system  and design a (simple) routing algorithm. Both papers use a {\it realizer} as defined in \cite{schnyder}. In the setting of unit Disk Graph (UDG) it is shown in \cite{DBLP:conf/algosensors/HucJLR10} that Schnyder's characterization prove to be useful for planarizing and routing on the communication graph. In \cite{DBLP:conf/icpads/SamarasingheL12,DBLP:conf/sensornets/SamarasingheL14}\footnote{Notice that our definition of greedy routing is different than the one used in \cite{DBLP:conf/icpads/SamarasingheL12,DBLP:conf/sensornets/SamarasingheL14}} the classical greedy-face routing paradigm is designed using coordinate similar to the ones used in the present article, see Figure \ref{fig:exvrac}, in the setting of UDG.

\noindent Our construction of the greedy routing algorithm on planar triangulations is summarized in the next Theorem. The pseudo-code of the algorithm is provided in Algorithm \ref{fig:speudocode} and the correctness of the algorithm is proved in the remarks \ref{remarkeven} and \ref{remarkodd} of the construction of the path if $D\in s^u_{2i}$ or $D\in s^u_{2i+1}$ that follow the Propositions \ref{prop:greedyeven} and \ref{lemmagreedy}.
\begin{theorem} Let us consider two nodes $u$ and $D$ of a planar triangulations. Using the coordinate system in Definition \ref{def:coordinates} there exists a greedy path between the two nodes. Moreover, the routing algorithm is local, only the coordinate of the destination and the neighboring nodes are necessary.
\end{theorem}


\begin{algorithm}[t]
\caption{Pseudo-code of the greedy routing}\label{fig:speudocode}
\begin{algorithmic}[1]
\State \textsc{input} Source $u$, Destination $D$
\Repeat
\If {$D\in {\cal N}_u$} $u=D$\Comment ${\cal N}_u$ is the set of neighbors of $u$
\Else
\If  { $D\in s^u_{2i-1}$}
	\State $u=v\in s^u_{2i-1}$ s.t. $(u,v)\in E$ \Comment $v$ is unique \vspace{0.1cm}
\Else \Comment  $D\in s^u_{2i}$ consider $v\in s^u_{2i-1}$ and $w\in s^u_{2i+1}$ s.t. $(u,v), (u,w)\in E$
	\If { $D\in s^v_{2i}$} 
		\State $u=v$ 
	\Else 
		\If{$D\in s^w_{2i}$} 
			\State $u=w$
		\Else \State $u=x\in s^u_{2i}$ s.t. $D\in s^x_{2i}$ \Comment must exist by Proposition \ref{prop:greedyeven}
		\EndIf
	\EndIf
\EndIf
\EndIf
\Until{u=D}
\end{algorithmic}
\end{algorithm}
\vspace{-0.2cm}

\section{Conclusion}
It this article we provide a definition of greedy routing that is independent of any graph embedding in metric space. Moreover, we use a new coordinate system such that greedy routing guarantees delivery and again without reference to any embedding of the graph. Besides the theoretical relevance, an important motivation for this article is to make geographic routing a real practical solution. This is the reason for avoiding the computation in the localization phase of the nodes. Moreover, the coordinate system that we use in the present paper uses only 3 integers and requires  ${\cal O}\bigl(log(n)\bigr)$ bits and can be qualified of {\it succint} \cite{eppstein2011succinct}. We emphasize that our algorithm does not require to embed the graph in a metric space. The next step towards the development of a general practical routing algorithm is to provide a distributed algorithm that extract  the standard representation of the communication graph (or a subgraph). Such an algorithm would generalize the approach in  \cite{DBLP:conf/icpads/SamarasingheL12,DBLP:conf/sensornets/SamarasingheL14} where the coordinates are measured and the techniques could be merged to overcome the situation where the graph is not maximal.

From a theoretical point of view, the result presented in the present article asks the question whether there exists graphs that admit greedy routing but no greedy embeddings. 

We have implemented our algorithm and simulated the routing process for some random planar triangulations. We obtained the random networks by placing the three distinguished nodes $A_1, A_2, A_3$ on the top of an equilateral triangle on the plane and generated nodes at random inside the triangle. The three order relations are defined by measuring the distances from the nodes to the distinguished nodes, i.e. $u<_i v\Leftrightarrow d(A_i,u)>d(A_i,v)$, $i=1,2,3$ where $d$ is the Euclidean distance on the plane, see the left of Figure \ref{fig:VRAC}. The equilateral triangle property implies that the representation is standard and constraining the nodes inside the triangle implies the Schnyder's empty property, see Property $a)$ of $(\ref{eq:schnydercond})$, \cite{DBLP:conf/sensornets/SamarasingheL14}. Then, the edges are generated in order to satisfy the condition $b)$ of $(\ref{eq:schnydercond})$. This is easily done by connecting each node $v$ to the three nodes  $v_1\in s^v_1, v_2\in  s^v_3,$ and $v_3\in  s^v_5$ such that $v_{i}=min_i\{ z\mid z\in s^v_{2i-1}\}$, $i=1,2,3$, accordingly to Property \ref{prop:sectors}. An instance of the random graphs is presented on Figure \ref{fig:experiment}.

 Our experiments validated the algorithm positively. The code and more complete explanations can be obtained upon request.

To conclude, we ask the question whether all planar triangulations can be drawn in the plane by distinguishing three nodes  $A_1, A_2, A_3$ and placing the other nodes on the plane inside the triangular region delimited by  $A_1, A_2, A_3$ \footnote{Similarly to how we produced random planar triangulations.} in such a way that the Schnyder orders are similar to the orders we obtain by measuring the distances from the nodes to the distinguished nodes. Answering positively would lead to a description (or even an algorithm) of a greedy embedding for planar triangulations.

\begin{figure}[!]\hskip -2cm
\hspace{2cm}
\centering
\includegraphics[scale=0.2]{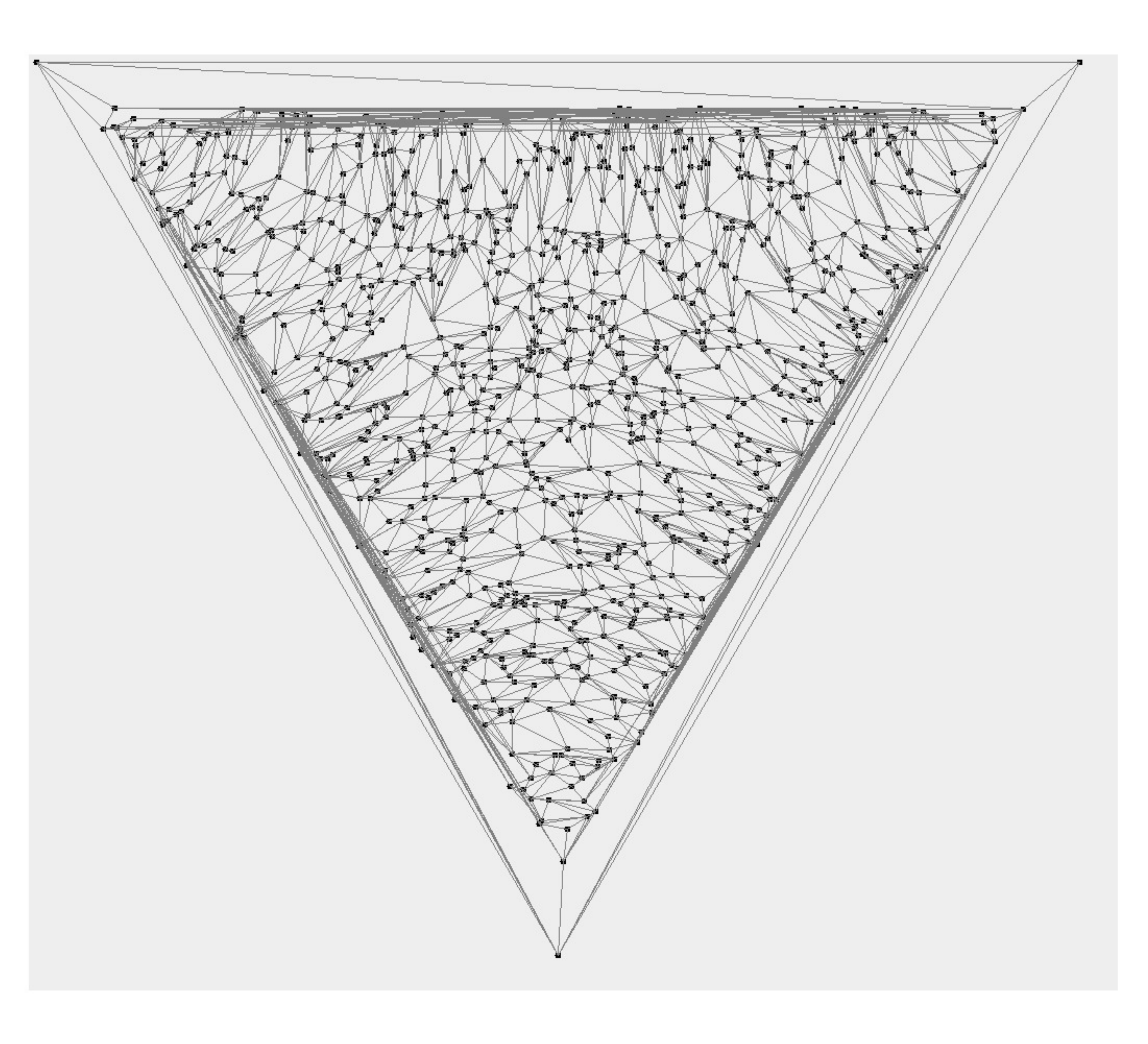}
\caption{Triangular Graph Generated by Experiments}
\label{fig:experiment}
 \end{figure}

\bibliographystyle{elsarticle-num} 
\bibliography{wsn,biblo}
\end{document}